\documentclass[11pt,a4paper]{article}

\usepackage[utf8]{inputenc}
\usepackage[T1]{fontenc}
\usepackage{lmodern}
\usepackage{amsmath,amssymb,amsthm}
\usepackage{array}
\usepackage{booktabs}
\usepackage{hyperref}
\usepackage{geometry}
\usepackage{color,soul}
\usepackage{yhmath}
\newtheorem{theorem}{Theorem}[section]
\newtheorem{proposition}[theorem]{Proposition}
\newtheorem{corollary}[theorem]{Corollary}

\newtheorem{definition}[theorem]{Definition}

\newtheorem{assumption}{Assumption}

\title{Evolution as a Process of Causal Inference}
\author{Jacopo Iacovacci\\[1ex]
\small Department of Epidemiology and Data Science\\ \small Fondazione IRCCS Istituto Nazionale dei Tumori di Milano\\ \small Via Giacomo Venezian 1, 20133, Milan, Italy\\[0.5ex]
\small jacopo.iacovacci@istitutotumori.mi.it}
\date{28 May 2026}

\begin{document}

\maketitle

\begin{abstract}
Recently, the mapping of the replicator equation onto Bayes' theorem has been recognised, leading to an analogy between evolutionary dynamics and Bayesian learning. 
However, this analogy holds only for pure selection in infinite populations and breaks down when mutations --- a central mechanism of evolution --- are introduced.

Here I propose that evolution by natural selection, at least for populations of haploid replicators in static environments, is best understood not as a learning process but as a process of causal inference.
Each mutation event constitutes a natural experiment in which the parent serves as the control and the mutant offspring as the treated unit.
Natural selection screens the causal effect of the mutation on fitness, retaining mutations with non-negative effects.
I formalise this view within the Neyman--Rubin potential--outcomes framework. I first develop the general theory using a generic fitness outcome and show how the core identification assumptions in causal inference (Stable Unit Treatment Value Assumption, Consistency, Unconfoundedness, Positivity) map onto evolutionary biology.

Using the unnormalised quasispecies equation, I prove that the intergenerational change in mean fitness decomposes exactly into a selection term --- recovering Fisher's Fundamental Theorem --- plus a mutation term that corresponds to a fitness--weighted average of the cumulated effect of all mutations over all parental genotypes. I show that this decomposition extends, under suitable assumptions, to the generalised replicator--mutator equation and that the frequencies of populations of matched parents--offspring update in proportion to the average causal effect of mutations on fitness.

Finally, I classify the major definitions of fitness according to whether they satisfy the Stable Unit Treatment Value Assumption (Appendix~A) and I list the major mutation mechanisms which guarantee Unconfoundedness (Appendix~B).

\medskip
\end{abstract}

\section{Introduction}
\label{sec:intro}

In 2009, Harper (Harper2009) and Shalizi (Shalizi2009) independently recognised a formal algebraic equivalence between the discrete replicator equation
with $n$ populations $\mathbf{x} = (x_1, \ldots, x_n)$ of fitnesses $w_1, \ldots, w_n$:
\begin{equation}
x_i' = \frac{x_i\, w_i}{\bar{w}}, \qquad \bar{w} = \sum_k x_k\, w_k,
\label{eq:replicator}
\end{equation}
and the Bayes' theorem for updating a prior $P(H_i)$ on hypotheses $H_1, \ldots, H_n$ after observing evidence $E$:
\begin{equation}
P(H_i \mid E) = \frac{P(E \mid H_i)\, P(H_i)}{P(E)}, \qquad P(E) = \sum_k P(E \mid H_k)\, P(H_k).
\label{eq:bayes}
\end{equation}
It is immediate to recognise the structural identity under the mapping reported in the following table
\begin{center}
\begin{tabular}{lll}
\toprule
Replicator & & Bayes \\
\midrule
Type frequency $x_i$ & $\longleftrightarrow$ & Prior $P(H_i)$ \\
Fitness $w_i$ & $\longleftrightarrow$ & Likelihood $P(E \mid H_i)$ \\
Mean fitness $\bar{w}$ & $\longleftrightarrow$ & Marginal likelihood $P(E)$ \\
Updated frequency $x_i'$ & $\longleftrightarrow$ & Posterior $P(H_i \mid E)$ \\
\bottomrule
\end{tabular}
\end{center}
This equivalence led to the proposal that evolution can be regarded as a process of Bayesian inference, where populations of replicators ``learn'' about the environment by updating their genotype distribution in proportion to their fitness, just as a Bayesian agent updates its hypotheses in proportion to the likelihood of evidence. The Bayesian analogy was later extended by Watson and Szathm\'{a}ry (Watson \& Szathm\'{a}ry, 2016), who argued that the capacity for learning increases with the complexity of the evolving system, and by Cz\'{e}gel et al.\ (Cz\'{e}gel et al. 2019), who interpreted multilevel selection as hierarchical Bayesian inference. More recently, Bettencourt et al.\ (Bettencourt et al., 2025) proposed redefining fitness as the likelihood function that emerges from this mapping.

While elegant, however, the Bayesian analogy suffers from limitations. First, it applies only to the pure replicator equation, i.e., to evolutionary processes of selection without mutation and in infinite populations.
Aky{\i}ld{\i}z (Aky{\i}ld{\i}z, 2017) recently showed that the replicator--mutator equation can not be written as a Bayesian update, and corresponds instead to the prediction step of a Hidden Markov Model.
Second, as Shalizi himself stressed (Shalizi, 2009), the correct direction of the analogy is that ``Bayesian updating also follows the replicator equation'' --- not the reverse.
The replicator equation is indeed a general description of proportional growth with normalisation and many dynamical systems share this form without being instances of Bayesian inference.
Third, the analogy inherits all the limitations of adaptationism (Gould \& Lewontin, 1979), ignoring neutral evolution, genetic drift, and non-adaptive forces that shape the majority of molecular evolutionary change (Lynch, 2007).

In this work, a fundamentally different perspective is proposed and its {mathematical foundation is laid down}. Evolution is regarded as a process of \emph{causal inference} (Pearl, 2010)
where mutation events act as natural interventions applied to single evolvable units.
Natural selection evaluates the causal effect of the mutation (intervention) on fitness and, consequently, populations evolve by accumulating interventions
with non-negative causal effects.

Throughout this work the focus is on haploid, asexual populations in static environments with point mutations, a setting in which the causal assumptions required for identification are most transparent.
In Section~\ref{sec:causal}, mutation events are formalised as individual-level interventions within the Neyman--Rubin potential--outcomes framework (Rubin, 1974), and the \emph{individual mutation effect} and the \emph{average mutation effect} on fitness are defined accordingly.
Section~\ref{sec:dynamics} embeds these causal quantities into classical models of evolutionary dynamics, yielding an exact decomposition of inter-generational fitness change into selection and mutation components for the quasispecies equation and identifying conditions under which the causal decomposition holds exactly or approximately for the generalised replicator--mutator equation. 
Section~\ref{sec:updating} clarifies the precise sense in which mutation disrupts the Bayesian--learning analogy, by showing that in the replicator--mutator dynamics, the updating of frequencies is proportional to the \emph{average mutation effect}.
In Section~\ref{sec:discussion}, I discuss the main conceptual and mathematical contributions of the proposed framework as well as its limitations and possible extensions, also in relation to existing causal approaches to evolution.
In Appendix~\ref{sec:sutva}, I provide a systematic classification of existing fitness definitions and classify them by whether they satisfy the Stable Unit Treatment Value Assumption, showing that individual-level definitions (viability, Wrightian under density independence, Malthusian, geometric mean, Euler--Lotka) satisfy SUTVA, while relational definitions (frequency-dependent, inclusive, invasion, game-theoretic, group) require interference frameworks.
Finally, in Appendix~\ref{sec:ignorability}, I discuss the main mechanisms of mutations which satisfy the Unconfoundedness assumption.

\section{Evolution as Natural Causal Inference}
\label{sec:causal}
\subsection{Mutation Events as Natural Interventions}

Consider a haploid organism (the parent) with genotype $g$, living in environment $E$. The vector $\mathbf{x}_g$ describes its genetic background.
The parent reproduces by generating an offspring which may carry a point mutation at a given locus due to the replication process. We define the mutation variable $M$:
\[
M \in \{0, 1\},
\]
where $M = 0$ indicates no mutation (i.e., the offspring is an exact copy of the parent) and $M = 1$ indicates a mutation.
The probability of mutation is $P(M = 1) = \mu$, corresponding to the per--generation mutation rate.

\begin{definition}[Potential Fitness Outcomes]
\label{def:potential}
For a given parent with genetic background $\mathbf{x}_g$ in environment $E$, let $Y$ denote the fitness of the hypothetical offspring it would produce.
We define two potential fitness outcomes for this hypothetical offspring given the mutation $M$ would occur or not:
\begin{align}
Y(0) &= {\text{fitness of offspring if no mutation occurred}} \\
Y(1) &= {\text{fitness of offspring if mutation occurred}}.
\end{align}
\end{definition}
$Y(1)$ and $Y(0)$ correspond to the values of an outcome of interest (in this case, the fitness) that a unit (the individual offspring) \emph{would have} under each intervention level (i.e., mutation or no mutation, respectively).
The fundamental problem of causal inference (Holland, 1986) is that for any unit, only one potential outcome can be observed:
\begin{equation}
Y_i^{\text{obs}} = M \cdot Y(1) + (1 - M) \cdot Y(0).
\label{eq:fundamental}
\tag{3}
\end{equation}

However, in evolution, nature provides a solution to this problem: when a haploid parent replicates and generates an offspring with a genetic mutation, parent and offspring constitute a \emph{natural matched pair} in which the offspring shares the genetic background of the parent and inhabits the same environment $E$. The parent's fitness $Y^{\text{par}}$ corresponds to $Y(0)$ (the fitness the offspring would have had without mutation) and the mutant offspring's fitness $Y^{\text{off}}$ corresponds to $Y(1)$. Therefore, the difference between the offspring fitness and the parent fitness can be regarded as the Individual Treatment Effect of the mutation.\\
We can therefore estimate the \emph{Individual Mutation Effect} (IME) of a mutation as follows:
\begin{equation}
\delta = Y^{\text{off}} - Y^{\text{par}} = Y(1) - Y(0) = \mathrm{IME}.
\label{eq:ime}
\tag{4}
\end{equation}
\noindent The Individual Mutation Effect (IME) measures the causal effect of the mutation on fitness by comparing the fitness of the mutant offspring to that of its parent, who serves as the natural, genetically matched control.
In Eq.~\eqref{eq:ime}, the fitness of the parent and of its offspring are evaluated over the same time window, starting at the time $T_0$ when the offspring is generated. This ensures that $Y^{\text{par}}$ and $Y^{\text{off}}$ represent outcomes over a shared temporal baseline, as required for a well-defined individual causal contrast.
\subsection{Average Mutation Effect}

Given Eq.~\eqref{eq:ime}, it is possible to define an \emph{Average Mutation Effect} as the expected value of the IME over a population of matched parent--offspring pairs in which all offspring have undergone the same type of mutation.

\begin{definition}[Average Mutation Effect]
\label{def:ame}
Over a population of parent--offspring pairs in a static environment where the same mutation has occurred in the offspring, the average causal effect of the mutation on fitness is defined as
\begin{equation}
\mathrm{AME} = \mathbb{E}[\delta_i] = \mathbb{E}\bigl[Y_i(1) - Y_i(0)\bigr].
\label{eq:ame_def}
\tag{5}
\end{equation}
\end{definition}

The AME corresponds to the Average Treatment Effect (ATE) of the Neyman--Rubin framework, when considering a mutation as the intervention and the fitness as the outcome under analysis. In particular, the AME is defined conditional on the parental genetic background and environment.\\
The standard assumptions of the potential--outcome framework helps understanding which evolutionary conditions are necessary to identify the causal mutation effects in Eqs.~\eqref{eq:ime} and~\eqref{eq:ame_def}.

\begin{assumption}[SUTVA --- Stable Unit Treatment Value Assumption]
\label{as:sutva}
The treatment applied to one unit does not affect the potential outcomes for another unit, and there is only one version of each treatment level.
\end{assumption}
\noindent The no-interference component is satisfied when the fitness of the parent population is independent of both the density and the fitness of the offspring population and vice versa.
Appendix~\ref{sec:sutva} analyses which fitness definitions satisfy or violate SUTVA.

\begin{assumption}[Consistency]
\label{as:consistency}
If $M = 1$, then $Y_i^{\text{obs}} = Y_i(1)$.
\end{assumption}
\noindent This assumption is satisfied by definition, as the observed fitness of a mutant is its fitness under that mutation.
\begin{assumption}[Ignorability / Unconfoundedness]
\label{as:ignorability}
$\{Y_i(0), Y_i(1)\} \perp\!\!\!\perp M \mid \mathbf{X}$. The treatment assignment is independent of the potential outcomes, conditional on a set of observed pre-treatment covariates.
\end{assumption}

\noindent This assumption relates to the nature of the mutation considered. Point mutations are, to first approximation, \emph{random} with respect to their fitness effects, i.e.,
they happen irregardless of whether they will be beneficial, neutral, or deleterious. Therefore, point mutations can be considered in first approximation randomised interventions that satisfy
the stricter condition $\{Y_i(0), Y_i(1)\} \perp\!\!\!\perp M$.
Appendix~\ref{sec:ignorability} analyses which classes of mutations satisfy the ignorability condition, providing the conditioning set $\mathbf{X}$ required for each class.

\begin{assumption}[Positivity]
\label{as:positivity}
$0 < P(M = 1 \mid \mathbf{X}) < 1$ for all $\mathbf{X}$.
\end{assumption}
\noindent This assumption is satisfied whenever the mutation rate $\mu > 0$, which is universally the case in biological populations.\\

When the four assumptions above are satisfied, an estimand of the AME can be identified:

\begin{theorem}[Identification of the AME]
\label{thm:identification}
Under Assumptions~\ref{as:sutva}--\ref{as:positivity}, the Average Mutation Effect of a mutation $M$ on the fitness $Y$ of a population with given genetic background $\mathbf{x}_g$ and in a given environment $E$ is identified from observable quantities as:
\begin{equation}
\mathrm{AME} = \mathbb{E}[Y \mid M = 1] - \mathbb{E}[Y \mid M = 0]. \label{eq:identification}
\tag{6}
\end{equation}
\end{theorem}

\begin{proof}
Under ignorability, $\mathbb{E}[Y(1) \mid M = 1] = \mathbb{E}[Y(1)]$.
Therefore:
\begin{align*}
\text{AME} &= \mathbb{E}[Y(1) - Y(0)] \\
	   &= \mathbb{E}[Y(1)] - \mathbb{E}[Y(0)] \\
           &= \mathbb{E}[Y(1) \mid M=1] - \mathbb{E}[Y(0) \mid M=0] \\
           &= \mathbb{E}[Y \mid M=1] - \mathbb{E}[Y \mid M=0]. \qedhere
\end{align*}
\end{proof}

It follows that it is possible to estimate the AME on fitness for a sufficient number of fitness observations from $N$ matched parents and mutant offspring as
\begin{equation}
\widehat{\mathrm{AME}} = \frac{1}{N} \sum_{i=1}^{N} \bigl[Y^1_i - Y^0_i \bigr] = \frac{1}{N} \sum_{i=1}^{N} \bigl[Y^{\text{off}}_i - Y^{\text{par}}_i \bigr].
\label{eq:ame_est}
\tag{7}
\end{equation}

\subsection{The Multi--Arm Structure of the Mutation Process}

The binary treatment variable $M \in \{0,1\}$ records whether a specific mutation occurs. However, different types of mutations can occur. Different amino acid substitutions at the same locus, mutations at different loci, and mutations producing different alleles all constitute distinct interventions with potentially distinct causal effects on a given genetic background $\mathbf{x}_g$.

Assuming a static environment, the average mutation effect can be rewritten in a general form to accommodate multiple mutations (corresponding to multiple intervention levels) and different genetic backgrounds (or types).
Let $\mathcal{M}_j$ denote the set of possible mutation outcomes from type $j$, where each element $m \in \mathcal{M}_j$ specifies a particular mutational change (e.g., a transition from $j$ to some type $k \neq j$).
The AME on the fitness is:
\begin{equation}
\mathrm{AME}(m, j) \equiv \mathrm{AME}(m: j \rightarrow k) = \mathbb{E}\bigl[Y(m) - Y(0) \mid \text{type } j \bigr] = \frac{1}{N} \sum_{i=1}^{N} \bigl[Y^k_i - Y^j_i \bigr],
\label{eq:ame_multi}
\tag{8}
\end{equation}
where the index $i$ runs over all matched parent--offspring pairs.

The mutation process from a parent population of type $j$ naturally gives rise to a multi--arm treatment structure. The mutation matrix provides the probabilities: $p(j \to k \mid j) = p_{jk}$, where $p_{jk}$ is the conditional probability of becoming type $k$ given that a mutation from type $j$ has occurred, with $p_{jj} = 0$ and $\sum_{k \neq j} p_{jk} = 1$. The mutation process from parent type $j$ can therefore be regarded as a \emph{randomised multi--arm treatment protocol}, where nature randomly assigns one of the treatments $m \in \mathcal{M}_j$ with probability $p_{mj} = p_{jk},\; m: j \rightarrow k$. Therefore, each distinct mutation $m$ corresponds to a distinct arm of a multi--arm trial.

For the multi--arm formulation, the identification assumptions must hold. The no--multiple--versions component of the SUTVA requires that each mutation type $m$ produces a unique fitness outcome regardless of when or in which individual it occurs. This is satisfied under a constant fitness landscape where all individuals of the same type share the same fitness. Consistency also requires \emph{well--defined interventions}, i.e., that mutations under study are uniquely defined.
Finally, positivity requires $p_{jk} > 0$ for all accessible mutation destinations $k$, which is satisfied when the mutation matrix has no structurally forbidden transitions within the accessible mutation set.

\section{Causal Effects in Models of Evolutionary Dynamics}
\label{sec:dynamics}

The causal inference nature of evolution can be identified in classical models of evolutionary dynamics.
Specifically, it can be shown that the inter--generational change in mean fitness decomposes into a selection component and a mutation component, where the latter is expressed in terms of the AME.

\subsection{The Unnormalised Discrete Quasispecies Equation}

Consider a haploid population with $n$ types.
Let $n_j(t)$ denote the number of individuals of type $j$ at generation $t$, $w_j$ the absolute (Wrightian) fitness of type $j$, and $Q_{jk}$ the mutation matrix describing the fraction of offspring from parent $j$ that become type $k$ during each replication event. The dynamics in absolute population counts are:
\begin{equation}
n_k(t+1) = \sum_j Q_{jk}\, w_j\, n_j(t).
\label{eq:quasi}
\tag{9}
\end{equation}
The total population is
\[
N(t) = \sum_j n_j(t),
\]
and the mean fitness is
\[
\bar{w}(t) = \frac{1}{N(t)} \sum_j w_j\, n_j(t).
\]
Summing Eq.~\eqref{eq:quasi} over $k$ and using $\sum_k Q_{jk} = 1$ gives the total population at time $t+1$:
\begin{equation}
N(t+1) = \sum_j w_j\, n_j(t) = \bar{w}(t)\, N(t).
\label{eq:Ntot}
\tag{10}
\end{equation}
The mutation matrix in Eq.~\eqref{eq:quasi} can be decomposed as follows:
\begin{equation}
Q_{jk} = \delta_{jk}(1 - \mu_j) + \mu_j\, q_{jk},
\label{eq:Qdecomp}
\tag{11}
\end{equation}
where $\mu_j$ is the total mutation rate away from type $j$, $q_{jk}$ is the conditional probability of becoming type $k$ given a mutation from $j$ (with $q_{jj} = 0$, $\sum_{k \neq j} q_{jk} = 1$), and $\delta_{jk}$ is the Kronecker delta. The expected fitness of the offspring of parent population type $j$, accounting for any mutation $m \in \mathcal{M}_j,\; m: j \rightarrow k$, can be written as
\begin{align}
w_j^* &= \sum_k w_k\, Q_{jk} = w_j(1 - \mu_j) + \mu_j \sum_{k \neq j} q_{jk}\, w_k \notag\\
      &= w_j + \mu_j \Bigl(\sum_{k \neq j} q_{jk}\, w_k - w_j\Bigr) \notag\\
      &= w_j + \mu_j \Bigl(\sum_{k \neq j} q_{jk}\, (w_k - w_j)\Bigr) \notag\\
      &= w_j + \mu_j \Bigl(\sum_{k \neq j} q_{jk}\, \mathrm{AME}(j \rightarrow k)\Bigr) \notag\\
      &= w_j + \sum_{m \in \mathcal{M}_j} p_{mj}\, \mathrm{AME}(m, j).
\label{eq:wjstar}
\tag{12}
\end{align}

The final term corresponds to the overall average mutation effect on fitness from any mutation originating in the parent population of type $j$. It is the weighted average of the AMEs from all matched mutant offspring--parent populations of parent type $j$, with weights equal to the corresponding mutation rates.

\begin{definition}[Total Mutation Effect]
\label{def:tme}
In a static environment and under a constant fitness landscape, the Total Mutation Effect on fitness for a parent population of type $j$, accounting for any mutation $m \in \mathcal{M}_j$, is defined as
\begin{equation}
\tau_j = \sum_{m \in \mathcal{M}_j} p_{mj}\, \mathrm{AME}(m, j),
\label{eq:tme_def}
\tag{13}
\end{equation}
where $p_{mj}$ is the fraction of the mutant offspring carrying mutation $m$ and $\mathrm{AME}(m, j)$ is the average mutation effect given mutation $m$ and parental type $j$.
\end{definition}

\subsection{Causal Decomposition of Mean Fitness Change}

\noindent Intuitively, the Total Mutation Effect $\tau_j$ summarises, for each parent type $j$, the average causal effect on fitness of all mutations that can arise from $j$, weighted by their probabilities.
Aggregating these individual-level mutation effects across the whole population yields the mutation component of the change in mean fitness, which adds to Fisher's selection term.\\
The following theorem can be proved:
\begin{theorem}[Causal decomposition of $\Delta\bar{w}$]
\label{thm:decomposition}
Under a constant fitness landscape, the change in mean fitness between generation $t$ and $t+1$ is:
\begin{equation}
\Delta\bar{w} \;=\; \frac{\mathrm{Var}(w)}{\bar{w}} \;+\; \frac{1}{\bar{w}}\,\mathbb{E}[w\,\tau],
\label{eq:decomp}
\tag{14}
\end{equation}
where $\mathrm{Var}(w)$ is the fitness variance in the population, and $\tau$ is the Total Mutation Effect between generations. All fitness values $w_j$ and Total Mutation Effects $\tau_j$ are
evaluated at the parent generation $t$.
\end{theorem}

\begin{proof}
Under a constant fitness landscape:
\begin{align*}
\bar{w}(t+1) &= \frac{1}{N(t+1)} \sum_j w_j\, n_j(t+1) = \frac{1}{\bar{w}\, N} \sum_j w_j\, n_j \Bigl(\sum_k w_k\, Q_{jk}\Bigr) = \frac{1}{\bar{w}} \sum_j x_j\, w_j\, w_j^*,
\end{align*}
where $x_j = n_j/N$. Substituting $w_j^* = w_j + \tau_j$ from Eq.~\eqref{eq:wjstar} gives
\begin{align*}
\bar{w}(t+1) &= \frac{1}{\bar{w}} \sum_j x_j\, w_j (w_j + \tau_j) = \frac{1}{\bar{w}} \Bigl(\sum_j x_j\, w_j^2 + \sum_j x_j\, w_j \tau_j\Bigr) = \frac{1}{\bar{w}} \Bigl(\mathbb{E}[w^2] + \mathbb{E}[w\,\tau]\Bigr).
\end{align*}
Therefore
\begin{align*}
\Delta\bar{w} &= \bar{w}(t+1) - \bar{w}(t) = \frac{\mathbb{E}[w^2] - \bar{w}^2}{\bar{w}} + \frac{\mathbb{E}[w\,\tau]}{\bar{w}} = \frac{\text{Var}(w)}{\bar{w}} + \frac{\mathbb{E}[w\,\tau]}{\bar{w}}. \qedhere
\end{align*}
\end{proof}

The first term, $\text{Var}(w)/\bar{w}$, is the \emph{selection component} of the inter-generation average fitness variation. It is non-negative and equals zero only when all parental types have identical fitness.
The second term, $\mathbb{E}[w\,\tau]/\bar{w}$, is the \emph{mutation component} of the same quantity. It corresponds to the ratio between the fitness-weighted average of the Total Mutation Effects over all parental types and the average parental fitness.

It is easy to recognise that Theorem~\ref{thm:decomposition} corresponds to the Price equation with the character $z$ set equal to the fitness $w$. Indeed, the standard Price equation for the change in the population mean of any character $z$ is:
\begin{equation}
\bar{w}\,\Delta\bar{z} = \text{Cov}(w, z) + \mathbb{E}[w\,\Delta z],
\label{eq:price}
\tag{15}
\end{equation}
where the first term is the selection covariance and the second is the transmission-bias term.

Setting $z = w$ and noting that $\Delta z_j = w_j^* - w_j = \tau_j$ gives
\[
\bar{w}\,\Delta\bar{w} = \text{Var}(w) + \mathbb{E}[w\,\tau],
\]
which is Eq.~\eqref{eq:decomp} multiplied by $\bar{w}$.
The Price equation is a mathematical identity that holds without causal assumptions (Doebeli et al., 2017). It partitions the change in a population mean into a covariance term and a remainder, but assigns no mechanistic interpretation to either.

Here we have implicitly demonstrated that, under the assumptions of (i) constant fitness landscape (i.e., fitness values $w_j$ not changing between the parent and offspring generations) and (ii) causal identifiability discussed in Section~\ref{sec:causal}, the transmission-bias term $\mathbb{E}[w\,\Delta z]$ acquires a causal meaning and can be interpreted as the causal effect of the mutation process on the trait variation between generations, rather than a statistical residual.

\begin{corollary}[Decomposition under uniform mutation rate]
\label{cor:uniform}
If $\mu_j = \mu$ for all $j$:
\begin{equation}
\Delta\bar{w} = \frac{\mathrm{Var}(w)}{\bar{w}} + \mu\,\mathbb{E}[\gamma] + \frac{\mu}{\bar{w}}\,\mathrm{Cov}(w,\gamma),
\label{eq:corollary}
\tag{16}
\end{equation}
where $\gamma_j = \sum_{k\neq j} q_{jk}\,\mathrm{AME}(j\to k)$ is the average mutation effect per mutation event for parent type $j$, and $\mathbb{E}[\gamma] = \sum_j x_j \gamma_j$ is the population-average of the same quantity.
\end{corollary}

\begin{proof}
With $\mu_j = \mu$ we have $\tau_j = \mu\,\gamma_j$. Hence
\[
\frac{1}{\bar{w}}\mathbb{E}[w\,\tau] = \frac{\mu}{\bar{w}}\mathbb{E}[w\,\gamma] = \mu\,\mathbb{E}[\gamma] + \frac{\mu}{\bar{w}}\,\mathrm{Cov}(w,\gamma).
\]
Substituting into Theorem~\ref{thm:decomposition} yields the claimed expression.
\end{proof}

When $\mu = 0$, the first term in the decomposition recovers the selection component of Fisher's Fundamental Theorem in the context of haploid clonal types with Wrightian fitness (where the total variance and additive genetic variance coincide), confirming that the mutation--dependent terms vanish appropriately in the selection-only limit. When there is no fitness variation ($\mathrm{Var}(w) = 0$), evolution is driven purely by the average causal effect of mutation, $\Delta\bar{w} = \mu\,\mathbb{E}[\gamma]$, showing that $\mathbb{E}[\gamma]$ quantifies the directional pressure mutation exerts on fitness in the absence of selection. At mutation--selection balance ($\Delta\bar{w} = 0$), the selective advantage of fitter types (measured by $\mathrm{Var}(w)/\bar{w}$) is exactly offset by the mutational load from deleterious mutations (measured by $-\mu\,\mathbb{E}[\gamma] - \frac{\mu}{\bar{w}}\mathrm{Cov}(w,\gamma)$), expressing a fundamental evolutionary principle: the equilibrium between selection's adaptive force and mutation's disruptive pressure.

The covariance $\mathrm{Cov}(w,\gamma)$ captures whether fitter genotypes have access to more or fewer beneficial mutations. When $\mathrm{Cov}(w,\gamma) > 0$, fitter types have more beneficial mutational neighbourhoods (positive evolvability--fitness correlation). When $\mathrm{Cov}(w,\gamma) < 0$, fitter types are surrounded by deleterious mutations (diminishing returns epistasis). If $\mathrm{Cov}(w,\gamma) = 0$, the mutational neighbourhood is independent of current fitness.

\subsection{The Error Threshold of the Quasispecies Theory}
Within the proposed causal framework, the error threshold of the quasispecies theory (Eigen \& Schuster, 1977) --- defined as a critical mutation rate $\mu_c$ above which selection cannot maintain the fittest genotype, causing the quasispecies population to delocalise across sequence space and lose the genetic information of the master sequence --- emerges naturally from mutation-selection balance. At equilibrium, $\Delta\bar{w} = 0$, so Theorem~\ref{thm:decomposition} yields
\[
\frac{\mathrm{Var}(w)}{\bar{w}} + \frac{1}{\bar{w}}\mathbb{E}[w\,\tau] = 0 \quad \Rightarrow \quad \mathrm{Var}(w) = -\mathbb{E}[w\,\tau],
\]
where $\tau$ denotes the Total Mutation Effect between generations (Definition~\ref{def:tme}) and $\mathbb{E}[w\,\tau]$ represents the fitness-weighted average of $\tau_j$ over parent types.

When mutations are predominantly deleterious ($\mathbb{E}[w\,\tau] < 0$), the selective advantage of fitter types (quantified by $\mathrm{Var}(w)/\bar{w}$) must counteract the mutational load (quantified by $-\mathbb{E}[w\,\tau]/\bar{w}$). As $\mu$ increases, $\mathbb{E}[w\,\tau]$ grows more negative until the selective advantage can no longer compensate --- defining the critical mutation rate $\mu_c$. This causal interpretation of the error threshold $\mu_c$  as marker of the point at which the fitness--weighted causal cost of mutations exactly cancels the selective benefit of fitness variance follows directly from the decomposition of Theorem~\ref{thm:decomposition} and finds no parallel in the Bayesian framework.

\subsection{Extension to the Generalised Replicator--Mutator Equation}
\label{sec:replicator_mutator}

Under certain conditions, the causal decomposition of mean fitness change can be extended to the replicator--mutator equation. Compared to Eq.~\eqref{eq:quasi}, the general form of the replicator--mutator equation includes both normalised population frequencies and frequency--dependent fitnesses.
The normalised quasispecies equation (with frequency--independent fitness $w_j$) can be written as:
\begin{equation}
x_i' = \frac{1}{\bar{w}} \sum_j Q_{ji}\, w_j\, x_j.
\label{eq:rm}
\tag{17}
\end{equation}
\noindent The mean fitness $\bar{w}(t) = \sum_k w_k x_k$ is determined entirely by the population state at generation $t$, i.e., before the mutation events producing the offspring.
This makes $\bar{w}(t)$ a pre--treatment quantity.

\begin{proposition}[Conditional SUTVA under normalisation]
\label{prop:cond_sutva}
Let the fitness values $w_j$ be frequency-independent.
Conditioning on the current population state $\mathbf{x}(t)$, the normalised quasispecies equation satisfies SUTVA.
\end{proposition}

\begin{proof}
Under frequency-independent fitness, $\bar{w}(t) = \sum_k w_k x_k(t)$ and all $x_k(t)$ are determined by the population at generation $t$.
Mutations occur during reproduction and affect only the offspring generation. Conditional on $\mathbf{x}(t)$, the denominator $\bar{w}(t)$ is fixed and does not vary with any individual offspring's mutation status.
\end{proof}

It follows that, under frequency-independent fitness, Eq.~\eqref{eq:decomp} holds identically for the normalised quasispecies equation. It is important to note that Proposition~\ref{prop:cond_sutva} establishes SUTVA \emph{within} a single generation, conditional on the current population state. Across generations, the offspring's mutation status at generation $t$ affects $\mathbf{x}(t+1)$, which in turn determines $\bar{w}(t+1)$, potentially creating inter--generational interference.
Here, each generation is treated as an independent cross--sectional causal experiment, identified conditional on the population state at that generation. This is the natural approach for discrete--generation models and is sufficient for the decomposition theorems, which concern single-generation changes. Accounting for a fully longitudinal intervention of sequential mutation events across multiple generations would require, in principle, time--varying treatment methods such as marginal structural models or g--computation (Robins, 1986; Hern\'{a}n \& Robins, 2006), which falls beyond the scope of this paper.\\

In the generalised replicator--mutator equation, fitness depends on the population composition:
\begin{equation}
x_i' = \frac{1}{\bar{w}(\mathbf{x})} \sum_j Q_{ji}\, w_j(\mathbf{x})\, x_j.
\tag{GRM}
\label{eq:grm}
\end{equation}
Therefore, a mutation in individual type $j$'s offspring that changes population composition also changes the fitness values of other types, resulting in a SUTVA violation.

However, when fitness varies slowly with composition and mutations are rare (weak frequency dependence) we have that
\[
w_j(\mathbf{x}') - w_j(\mathbf{x}) = O(\mu/N) \quad \text{per individual mutation event}.
\]

\begin{proposition}[SUTVA under weak frequency dependence]
\label{prop:approx_sutva}
Under frequency-dependent fitness with bounded sensitivity $\left\|\partial w_j / \partial x_k\right\| \leq C$ for all $j,k$ and uniform mutation rate $\mu$, the causal decomposition holds up to $O(\mu^2)$:
\[
\Delta\bar{w} = \frac{\mathrm{Var}(w(\mathbf{x}))}{\bar{w}(\mathbf{x})} + \frac{1}{\bar{w}(\mathbf{x})}\,\mathbb{E}\bigl[w(\mathbf{x})\,\tau(\mathbf{x})\bigr] + R(\mu^2),
\]
where:
\begin{itemize}
\item All fitness values $w_j(\mathbf{x})$ and Total Mutation Effects $\tau_j(\mathbf{x})$ are evaluated at the parent-generation state $\mathbf{x} = \mathbf{x}(t)$
\item $\tau_j(\mathbf{x}) = \sum_{m \in \mathcal{M}_j} p_{mj} \,\mathrm{AME}(m,j;\mathbf{x})$ is the Total Mutation Effect for parent type $j$, with $\mathrm{AME}(m,j;\mathbf{x})$ computed under the fixed fitness landscape $w_j(\mathbf{x})$
\item $p_{mj} = \mu_j q_{jk}$ is the fraction of offspring from type $j$ that are mutants of type $k$ (where $m: j \rightarrow k$)
\end{itemize}
The remainder satisfies
\begin{equation}
|R(\mu^2)| \;\leq\; \frac{n\, C\, \mu^2\, w_{\max}}{\bar{w}(\mathbf{x})}\,
\bigl(\max_j |\tau_j(\mathbf{x})|\bigr),
\label{eq:error_bound}
\end{equation}
with $n$ the number of types and $w_{\max} = \max_j (w_j)$.
\end{proposition}

\begin{proof}
Under the stated conditions, a first-order Taylor expansion of $w_j(\mathbf{x}')$ around $\mathbf{x}$ gives:
\[
w_j(\mathbf{x}') = w_j(\mathbf{x}) + \sum_k \frac{\partial w_j}{\partial x_k}\bigl(x_k' - x_k\bigr) + O(\|\mathbf{x}' - \mathbf{x}\|^2).
\]
Since each individual mutation event shifts $x_k$ by $O(\mu/N)$ and there are $O(N)$ independent replication events, the aggregate frequency shift satisfies $\|\mathbf{x}' - \mathbf{x}\| = O(\mu)$ for each $k$. The first-order correction to $w_j$ is therefore bounded by $n \cdot C \cdot O(\mu)$.

In the proof of Theorem~\ref{thm:decomposition}, the term $\bar{w}(t+1) = \frac{1}{\bar{w}} \sum_j x_j w_j w_j^*$ uses $w_j^*$ evaluated at the parental generation fitness landscape. Substituting $w_j(\mathbf{x})$ for $w_j(\mathbf{x}')$ in computing $\bar{w}(t+1)$ incurs an error
\[
|R| \leq \frac{1}{\bar{w}} \sum_j x_j w_j \cdot |\text{correction to } w_j^*| \leq \frac{1}{\bar{w}} \cdot w_{\max} \cdot \mu \cdot n C \cdot \mu \cdot \max_j |\tau_j(\mathbf{x})|,
\]
resulting in the bound in Eq.~\eqref{eq:error_bound}. For biologically realistic parameters based on experimental evolution and mutation--rate studies with microbes (e.g. Drake, 1991, Lenski et al., 1991, Tenaillon et al., 2016) --- per--genome mutation rates in non-hypermutator \textit{E.\ coli} strains typically in the range $\mu \sim 10^{-4}$\text{--}$10^{-3}$ per generation, frequency-dependent selection strength $\left|\partial w_j / \partial x_k\right| \sim 10^{-2}$ to $10^{-3}$ per 0.01 change in genotype frequency, and number of distinct genotypes $n \sim O(10)$ --- we obtain, taking the conservative value $\mu \sim 10^{-3}$,
\[
|R| \;\leq\;
\frac{(10)\cdot(10^{-2})\cdot(10^{-3})^2\cdot w_{\max}}{\bar{w}}\,
\max_j |\tau_j|
\;\leq\; 10^{-7}\,\frac{w_{\max}}{\bar{w}}\,
\max_j |\tau_j|.
\]
Since $\frac{w_{\max}}{\bar{w}}$ and $\max_j |\tau_j|$ are typically $O(1)$ (fitness effects and mutation effects are rarely orders of magnitude larger than mean fitness), the remainder $|R|$ is at most $\sim 10^{-7}$ in absolute terms. Given that the leading--order terms $\frac{\mathrm{Var}(w)}{\bar{w}}$ and $\frac{1}{\bar{w}}\mathbb{E}[w\,\tau]$ are themselves $O(10^{-2})$ to $O(1)$ for evolving populations, this error is negligible --- about five to seven orders of magnitude smaller than the smallest expected signal in typical evolutionary dynamics.
\end{proof}
\section{Updating in Replicator--Mutator Dynamics}
\label{sec:updating}

In this section the precise nature of the deviation from Bayesian updating when mutation is included in the pure replicator equation is shown to be captured by the Average Mutation Effect.

\subsection{The HMM Structure of the Replicator--Mutator Equation}
\label{sec:hmm}

The replicator--mutator dynamics in Eq.~\eqref{eq:grm} can be decomposed into two sequential operations:

\paragraph{Step 1 --- Selection}
\begin{equation}
\tilde{x}_j = \frac{w_j\, x_j}{\bar{w}};
\label{eq:stepA}
\tag{18}
\end{equation}
this is a Bayesian update with prior $x_j$ and likelihood $w_j$. The resulting distribution $\tilde{\mathbf{x}}$ is the ``Bayesian posterior'' after observing one generation of selection.

\paragraph{Step 2 --- Mutation}
\begin{equation}
x_i(t+1) = \sum_j \tilde{x}_j\, Q_{ji};
\label{eq:stepB}
\tag{19}
\end{equation}
this is a Markov transition applied to the post-selection distribution, spreading probability mass from each type to its mutational neighbours.

We can also implement the same dynamics following a \emph{mutation-first, selection-second} operation (sometimes referred to as \emph{mutator--replicator} dynamics) by propagating frequencies through the mutation matrix (transition kernel) first and then weighting by fitness (likelihood):

\paragraph{Step 1 --- Mutation (Transition kernel)}
\begin{equation}
\tilde{x}_i = \sum_j Q_{ij}\, x_j(t);
\label{eq:mut_first}
\tag{20}
\end{equation}

\paragraph{Step 2 --- Selection (Bayesian update)}
\begin{equation}
x_i(t+1) = \frac{w_i\, \tilde{x}_i}{\tilde{\bar{w}}}, \qquad \tilde{\bar{w}} = \sum_k w_k\, \tilde{x}_k.
\label{eq:sel_second}
\tag{21}
\end{equation}

Aky{\i}ld{\i}z (Aky{\i}ld{\i}z, 2017) recognised that the replicator--mutator equations correspond to the \emph{prediction step} of a Hidden Markov Model (HMM) while the mutator--replicator equations correspond to a \emph{filtering step} of a HMM. Indeed, in the HMM framework, a hidden state $\theta_t$ evolves according to a transition kernel $K(\theta_t \mid \theta_{t-1})$, and at each time step an observation $e_t$ is emitted with likelihood $\phi(\theta_t) = P(e_t \mid \theta_t)$. The predictive step can then be written as:
\begin{equation}
P(\theta_t = i \mid e_{1:t-1}) = \sum_j P(\theta_t = i \mid \theta_{t-1} = j)\, P(\theta_{t-1} = j \mid e_{1:t-1}).
\label{eq:hmm_pred}
\tag{22}
\end{equation}
The filtering recursion --- which computes the posterior distribution of the hidden state given all observations --- is given by:
\begin{equation}
P(\theta_t = i \mid e_{1:t}) \propto P(e_t \mid \theta_t = i) \sum_j K(i \mid j)\, P(\theta_{t-1} = j \mid e_{1:t-1}).
\label{eq:hmm_filter}
\tag{23}
\end{equation}

\subsection{Causal Identification in the Mutator--Replicator Equation}
\label{sec:mutator_replicator}

The replicator--mutator dynamics can be implemented in two equivalent ways, either selection first followed by mutation, or mutation first followed by selection. While the dynamics are mathematically equivalent, the mutation-first ordering makes the causal structure more transparent as in nature the mutational intervention precedes the selective evaluation of its fitness consequences.
When this order is considered, it can be shown that the updating rule for mutant genotype frequencies admits a direct causal decomposition.

Consider a monomorphic parent population of type $j$ so that $x_j(t=0) = 1$ and all other $x_k(t=0) = 0$, with a small mutation rate $\mu_j \ll 1$.
For a specific mutant type $i \neq j$, let $q_{ji}$ denote the conditional probability of producing type $i$ given that a mutation occurs from $j$, so that the probability of a $j \to i$ mutation is $\mu_j q_{ji}$.
Since all offspring of type $j$ that mutate become type $i$ with probability $q_{ji}$, after the mutation step occurs we have
\[
\tilde{x}_j = 1 - \mu_j, \qquad \tilde{x}_i = \mu_j q_{ji}.
\]
After the selection step occurs we have:
\begin{equation}
x_i(t+1) = \frac{w_i\, \tilde{x}_i}{\tilde{\bar{w}}} = \frac{w_i\, \mu_j\, q_{ji}}{w_j(1-\mu_j) + w_i\, \mu_j\, q_{ji}},
\label{eq:xi_exact}
\tag{24}
\end{equation}
\begin{equation}
x_j(t+1) = \frac{w_j\, \tilde{x}_j}{\tilde{\bar{w}}} = \frac{w_j(1-\mu_j)}{w_j(1-\mu_j) + w_i\, \mu_j\, q_{ji}}.
\label{eq:xj_exact}
\tag{25}
\end{equation}

Because $\mu_j \ll 1$, to first order in $\mu_j$ we have $\tilde{\bar{w}} \approx w_j$ and therefore
\begin{equation}
\boxed{x_i(t+1) \approx \frac{w_i}{w_j}\, \mu_j\, q_{ji} \quad \text{to } O(\mu_j).}
\label{eq:xi_approx}
\tag{26}
\end{equation}
\begin{equation}
\boxed{x_j(t+1) \approx 1 - \mu_j\, \frac{w_i\, q_{ji}}{w_j} \quad \text{to } O(\mu_j).}
\label{eq:xj_approx}
\tag{27}
\end{equation}

\begin{proof}
To obtain the approximations we expand the denominators in Eqs.~\eqref{eq:xi_exact} and~\eqref{eq:xj_exact}.
For the offspring frequency:
\[
x_i(t+1) = \frac{w_i \mu_j q_{ji}}{w_j(1-\mu_j) + w_i \mu_j q_{ji}}
= \frac{w_i \mu_j q_{ji}}{w_j}\,
\frac{1}{1 - \mu_j + \frac{w_i}{w_j}\mu_j q_{ji}}.
\]
Since $\mu_j \ll 1$, we use $(1+\epsilon)^{-1} \approx 1-\epsilon$ with
$\epsilon = -\mu_j + \frac{w_i}{w_j}\mu_j q_{ji} = O(\mu_j)$, obtaining
\[
x_i(t+1) = \frac{w_i \mu_j q_{ji}}{w_j}\bigl(1 + O(\mu_j)\bigr)
= \frac{w_i}{w_j}\, \mu_j\, q_{ji} + O(\mu_j^2).
\]
For the parent frequency, we compute $1 - x_j(t+1)$ from Eq.~\eqref{eq:xj_exact}:
\[
1 - x_j(t+1) =
\frac{w_i\, \mu_j\, q_{ji}}{w_j(1-\mu_j) + w_i\, \mu_j\, q_{ji}}
= \frac{\frac{w_i}{w_j}\mu_j q_{ji}}{1 - \mu_j + \frac{w_i}{w_j}\mu_j q_{ji}}.
\]
Again the denominator is $1 + O(\mu_j)$, so
\[
1 - x_j(t+1) = \frac{w_i}{w_j}\mu_j q_{ji} + O(\mu_j^2),
\]
which yields Eq.~\eqref{eq:xj_approx}.
\end{proof}

We call the ratio $w_i/w_j$ the \emph{mutational fitness ratio} (MFR) for the mutation $j\to i$.
It expresses the effect of the mutation on fitness of the in a multiplicative form, whereas the Average Mutation Effect $\mathrm{AME}(j\to i)\approx w_i-w_j$ expresses the same effect additively.
Indeed, $w_i/w_j = 1 + \mathrm{AME}(j\to i)/w_j$, showing that the MFR is a rescaled version of the AME.

\subsubsection{Extension to $k$ Offspring Types}
For a parent population $j$ which can produce $k$ mutant types $i_1,\ldots,i_k$ with probabilities $Q_{i_l j} = \mu_j\, q_{j,i_l}$, with mutation rate with $\mu_j \ll 1$, the first-generation frequency of each mutant type is
\begin{equation}
\boxed{x_{i_l}(t+1) \approx \frac{w_{i_l}}{w_j}\, \mu_j\, q_{j,i_l} \quad \text{to } O(\mu_j).}
\label{eq:xi_multi}
\tag{28}
\end{equation}
The parent frequency is
\begin{equation}
\boxed{x_j(t+1) \approx 1 - \mu_j\, \sum_{l=1}^{k} \frac{w_{i_l}}{w_j}\, q_{j,i_l} \quad \text{to } O(\mu_j).}
\label{eq:xj_multi}
\tag{29}
\end{equation}

The proof follows the same steps as above, noting that after mutation the total mutant frequency is $\sum_l \mu_j q_{j,i_l} = \mu_j\sum_l q_{j,i_l} = \mu_j$ (since $\sum_l q_{j,i_l}=1$), and that selection acts independently on each mutant type to first order in $\mu_j$. By defining the population-averaged MFR, $\mathrm{\overline{MFR}}_j = \sum_l (w_{i_l}/w_j)\, q_{j,i_l}$, we can write the parent frequency compactly as 
\begin{equation}
x_j(t+1) \approx 1 - \mu_j\, \mathrm{\overline{MFR}}_j.
\tag{30}
\end{equation}

These expressions show that, to first order in the mutation rate, the mutator--replicator dynamics yields a frequency update in which each mutant type's excess frequency over the neutral expectation $\mu_j q_{j,i_l}$ is proportional to the mutational fitness ratio $w_{i_l}/w_j$, and hence to the Average Mutation Effect $\mathrm{AME}(j \to i_l)$. When these individual contributions are aggregated over all parent types and mutation outcomes, they generate the mutation term $\mathbb{E}[w\,\tau]/\bar{w}$ in Theorem~\ref{thm:decomposition}, so that the population--level change in mean fitness is expressed directly in terms of the same causal mutation effects that control the first--generation mutant frequencies.

\section{Discussion}
\label{sec:discussion}

This work began by challenging the recently--proposed view that evolution constitutes a process of Bayesian learning. This view originated from the elegant recognition of the algebraic equivalence between the pure replicator equation and Bayes' theorem. However, it was also recognised that the Bayesian update structure collapses when mutations are introduced in the replicator dynamics.

Here, an alternative view is developed where evolution by natural selection is assimilated to a process of causal inference. The key conceptual contribution is the identification of mutation events as individual--level natural experiments: when a haploid parent replicates and produces a mutant offspring, the parent--offspring pair constitutes a naturally matched control--treated unit, and the difference in fitness between them measures the individual effect of the mutation on fitness for the given parental genetic background and environment. In this fashion, mutations act at the population level as interventions whose average causal effect on fitness drive evolutionary dynamics.

Conceptually, the proposed causal--evolutionary framework complements Otsuka's structural interventionist analysis of evolutionary genetics (Otsuka, 2016). Otsuka developed causal models for how interventions on population--genetic parameters (e.g., fitnesses, recombination rates, or selection regimes) influence evolutionary outcomes, operating primarily at the level of hypothetical $\mathrm{do}$--operations on such parameters. By contrast, the present approach takes as its primitive causal events the mutational changes in individual organisms, treating them as natural experiments where mutations serve as interventions and fitness as the outcome measure. It formulates causal effects directly in terms of potential outcomes for individual parent--offspring pairs and then aggregates these to population-level decompositions. 

Crucially, the proposed theoretical framework is not founded at the scale of infinite populations of replicators and is not limited to fitness. For any heritable trait $z$, the intergenerational change in its population mean decomposes into a selection term $\mathrm{Cov}(w,z)/\bar{w}$ and a mutation term $\frac{1}{\bar{w}}\mathbb{E}[w\,\zeta]$, where $\zeta_j$ denotes the average mutation effect on $z$ for parent type $j$, allowing direct application to quantitative genetics and phenotypic evolution.  The mutation term $\frac{1}{\bar{w}}\mathbb{E}[w\,\zeta]$ provides a causal interpretation of the mutational input to phenotypic change that extends the classical additive--variance framework: rather than entering as a statistical residual, mutational pressure is now expressed as the fitness--weighted average of individual mutation effects on the trait, making its directional contribution to evolutionary change explicit and estimable from matched parent--offspring data.

The exact decomposition of Theorem~\ref{thm:decomposition} is derived for the deterministic (infinite population) quasispecies equation. In finite populations, genetic drift introduces stochastic fluctuations in mean fitness that are not captured by the causal mutation term alone. Extending the causal decomposition to finite populations would require accounting for the sampling variance of causal effects across generations, and represents a relevant direction for future work.

To develop the theory, the identification assumptions for the average causal effect of the mutations were analysed. Most importantly, the main fitness definitions satisfying SUTVA and the main mutation mechanisms fulfilling ignorability were reported in Appendix. These conditions hold for point mutations in haploid asexual populations inhabiting stable environments, which was the focus of this work. In this setting, the matched parent--offspring design is valid and fitness reflects individual traits without interference, allowing for a relatively simple mathematical treatment.
When these assumptions are relaxed --- such as in frequency--dependent fitness, rapidly changing environments, or the presence of epistasis and pleiotropy --- the proposed framework would require extension. Frequency--dependent fitness and environmental variation introduce interference and weaken the matched-pair design, necessitating the application of more advanced concepts of causal inference developed to address the presence of interference or time--varying treatments. Similarly, recombination, epistasis, and population structure would test the robustness of the causal interpretation, as they alter the relationship between individual mutational effects and population--level dynamics.

In conclusion, the proposed framework complements structural--causal analyses of selection and heredity by focusing on mutational interventions as primitive causal events, offering a complementary perspective that could be extended and integrated into a more unified causal theory of evolutionary dynamics.

\appendix

\section{Fitness Definitions and SUTVA}
\label{sec:sutva}

In the causal framework of Section~\ref{sec:causal}, fitness definitions are interpreted as functionals of an individual's reproductive output over a specified time window. For matched parent--offspring pairs, both fitnesses are evaluated over the same forward time interval starting at the offspring generation time $T_0$, so that they provide temporally aligned outcomes for the individual causal contrast $Y^{\text{off}} - Y^{\text{par}}$ defined in Eq.~(4). In this appendix I classify standard fitness concepts according to whether they satisfy SUTVA (Assumption~\ref{as:sutva}), thereby indicating when the identification theorem applies directly and when frameworks for causal inference with interference are required.

Throughout, SUTVA requires that each unit's potential outcomes $\{Y_i(0), Y_i(1)\}$ depend only on its own treatment status, and not on the treatments of other units. Fitness definitions that are strictly individual-level functionals of the life-history schedule and the external environment satisfy SUTVA, and Theorem~\ref{thm:identification} can be applied directly. Fitness definitions that depend on population composition, social partners, or group structure exhibit interference and require extended causal frameworks accounting for interference (Hudgens \& Halloran, 2008; Tchetgen Tchetgen \& VanderWeele, 2012).

\subsection{Definitions That Satisfy SUTVA}

The following fitness definitions assign a value to an individual based only on its own traits and the external environment, with no dependence on other units' genotypes or mutations.

\paragraph{Individual viability.}
Under pure viability selection, genotypic fitness is the probability of survival to reproductive age:
\[
V_g = P(\text{survive to reproduce} \mid \text{genotype } g).
\]
In models where survival to reproductive age is determined solely by the focal individual's genotype and the external environment, this viability is an individual-level quantity and satisfies SUTVA, since one individual's potential survival outcome does not depend on the mutation status of others.

\paragraph{Wrightian (absolute) fitness.}
Absolute fitness $w_i$ is the expected number of offspring contributed to the next generation by individual $i$. When $w_i$ is modelled as a function only of the focal individual's traits and the external environment, and not of the genotypes or mutation status of other units, each unit's potential outcomes depend only on its own treatment, and SUTVA is satisfied by construction.

\paragraph{Malthusian fitness.}
For continuously growing populations $N(t)$, Malthusian fitness is the per-capita growth rate $m$, where $N(t) = N(0)\, e^{mt}$. In a fixed environment with no density dependence, $m$ is a property of the individual life-history schedule and the external environment, so that the potential growth outcomes of one individual do not depend on others' mutation status. SUTVA holds in this setting.

\paragraph{Geometric mean fitness.}
In temporally fluctuating environments, the geometric mean fitness
\[
G = (w_1 w_2 \cdots w_T)^{1/T}
\]
is defined per genotype across time (Lewontin \& Cohen, 1969; Gillespie, 1973). It is constructed from the sequence of fitnesses experienced by the focal type across generations and does not involve interactions between individuals within a generation. Consequently, the potential values of $G$ for individual $i$ depend only on $i$'s own treatment status and the environmental sequence, and SUTVA is satisfied.

\paragraph{Euler--Lotka fitness.}
In age-structured populations, fitness is often defined as the intrinsic rate of increase $r$, found by solving the Euler--Lotka equation
\[
1 = \int_0^{\infty} l(x)\, m(x)\, e^{-rx}\, dx,
\]
where $l(x)$ is the survival function and $m(x)$ is age-specific fecundity (Charlesworth, 1994). In a fixed environment without frequency- or density-dependent interactions, $r$ is an individual-level functional of the life-history schedule $(l,m)$. Under such assumptions, an individual's potential $r$ depends only on its own treatment status, and SUTVA holds.

\subsection{Definitions That Violate SUTVA}

The following fitness definitions make the outcome of individual $i$ depend on the treatment status (genotype or mutation) of other individuals. When considering these fitness definitions, one should use frameworks for causal inference with interference, which define estimands and estimators that remain well-posed when each unit's potential outcomes may depend on others' treatments (Hudgens \& Halloran, 2008; Tchetgen Tchetgen \& VanderWeele, 2012). 

\paragraph{Frequency-dependent Wrightian fitness.}
When Wrightian fitness is frequency-dependent, we write $w_i = w(g_i, \mathbf{x})$, where $\mathbf{x}$ is the population composition (e.g. genotype frequencies). A mutation in individual $j$ that changes $\mathbf{x}$ also changes individual $i$'s fitness, so that $i$'s potential outcomes depend on $j$'s treatment; this constitutes interference and violates SUTVA. Experimental evolution studies increasingly report frequency-dependent fitness effects for beneficial mutations, indicating that such SUTVA violations are likely to be common for this class of fitness measures. Concrete examples include microbial competition assays in which the selection coefficient of a mutant changes with its initial frequency in the population.

\paragraph{Inclusive fitness.}
Hamilton's inclusive fitness (Hamilton, 1964) explicitly incorporates effects on relatives:
\[
w_{i}^{\text{incl}} = w_{i} + \sum_{j \in \text{relatives}} r_j \Delta w_j,
\]
where $r_j$ is a relatedness coefficient and $\Delta w_j$ is the mutation-induced change in relative $j$'s fitness. A mutation in individual $j$ that changes $j$'s behaviour toward $i$ alters $i$'s inclusive fitness directly. Thus, $i$'s potential inclusive-fitness outcome depends on $j$'s treatment status, violating SUTVA.

\paragraph{Invasion fitness.}
In adaptive dynamics (Metz et al., 1992), invasion fitness $f(x,y)$ is defined as the long-term growth rate of a rare mutant $y$ in a resident population at equilibrium $x$:
\[
f(x, y) = \lim_{T \to \infty} \frac{1}{T} \ln \frac{N_y(T)}{N_y(0)}.
\]
This is inherently a two-type quantity: the mutant's outcome depends on the resident's equilibrium state $x$, which in turn reflects the population's genetic composition. A change in the resident genotype $x$ due to other mutations changes the mutant's potential invasion fitness, violating SUTVA.

\paragraph{Game-theoretic fitness.}
In evolutionary game theory, fitness is derived from payoffs in strategic interactions, which depend on the strategies of interaction partners (Maynard Smith \& Price, 1973). A mutation in one individual alters the strategy profile and hence the payoff matrix experienced by others. Therefore, an individual's potential game-theoretic fitness depends on the mutation statuses of its partners, constituting interference and violating SUTVA.

\paragraph{Multilevel/group-selection fitness.}
In multilevel or group-selection models, group-level fitness is defined as a function of group composition, while individual fitness may have both within-group and between-group components (Wade, 1978). A mutation in one group member changes the group's composition and thus the group-level fitness of all members, as well as the partitioning of selection into within- and between-group components. Consequently, individuals' potential outcomes depend on others' treatment statuses, and SUTVA is violated.

\subsection{Summary}

\begin{table}[ht]
\centering
\caption{Classification of fitness definitions by SUTVA status.}
\label{tab:sutva}
\begin{tabular}{>{\raggedright\arraybackslash}p{0.30\linewidth}
                >{\centering\arraybackslash}p{0.12\linewidth}
                >{\raggedright\arraybackslash}p{0.45\linewidth}}
\toprule
Fitness definition & SUTVA & Causal framework needed \\
\midrule
Individual viability                  & Yes & Standard potential outcomes \\
Wrightian (density-independent)       & Yes & Standard potential outcomes \\
Malthusian fitness                    & Yes & Standard potential outcomes \\
Geometric mean fitness                & Yes & Standard potential outcomes \\
Euler--Lotka fitness                  & Yes & Standard potential outcomes \\
\midrule
Frequency-dependent fitness           & No  & Causal inference with interference \\
Inclusive fitness                     & No  & Causal inference with interference \\
Invasion fitness                      & No  & Causal inference with interference \\
Game-theoretic fitness                & No  & Causal inference with interference \\
Group/multilevel fitness              & No  & Causal inference with interference \\
\bottomrule
\end{tabular}
\end{table}

This classification shows that the choice of fitness definition determines which causal-identification strategy is appropriate. Where SUTVA is satisfied, the causal effect of a mutation is identified by the simple difference in conditional expectations, as given in Theorem~\ref{thm:identification}. Where SUTVA fails, the established literature on causal inference with interference provides the required extensions, including partial interference under group structure, exposure mappings, and inverse-probability-of-treatment weighting under interference (Hudgens \& Halloran, 2008; Tchetgen Tchetgen \& VanderWeele, 2012).

\newpage
\section{Mutation Classes and Ignorability}
\label{sec:ignorability}

The ignorability assumption (Assumption~\ref{as:ignorability}) requires that the occurrence of a mutation is independent of its fitness effect, conditional on pre-mutation covariates. In evolutionary terms, this means the mutation process must be ``blind'' to its phenotypic consequences: the probability that a particular mutation occurs does not depend on whether that mutation would be beneficial, neutral, or deleterious, once relevant covariates have been controlled. This assumption is what makes mutation a credible natural experiment --- a quasi--randomised intervention whose causal effect can be identified from observational data.

Here, I classify major mutational mechanisms according to whether they satisfy ignorability unconditionally, satisfy it only after conditioning on observable covariates, or violate it in ways that necessitate genuinely new causal structure. For mutation classes that do not satisfy ignorability unconditionally, the ``remedies'' listed below specify additional conditioning variables under which ignorability is restored, thereby defining extended causal estimands (conditional AMEs). These conditional AMEs can then be aggregated over the empirical distribution of the conditioning variables to recover the unconditional AMEs and Total Mutation Effects that enter the dynamical equations of Section~\ref{sec:dynamics}.

In causal-graph terms, mechanisms that induce ignorability violations do so by creating back-door paths from potential fitness outcomes to mutation assignment through covariates such as genomic context or stress state. Conditioning on these covariates blocks the back-door paths, restoring the independence of treatment assignment from potential outcomes in the standard sense. The resulting estimands are conditional AMEs of the form $E[Y(1) - Y(0) \mid X]$, where $X$ summarises the relevant covariate set.

Table~\ref{tab:ignorability} summarises the classification. In the remainder of this section, I provide the biological and inferential justification for each assignment, with an emphasis on how the remedies extend the scope of our causal--evolutionary framework.

\begin{table}[ht]
\centering
\caption{Classification of mutation classes by ignorability status. The final column lists conditions under which Assumption~\ref{as:ignorability} is restored.}
\label{tab:ignorability}
\begin{tabular}{>{\raggedright\arraybackslash}p{0.26\linewidth}
                >{\centering\arraybackslash}p{0.14\linewidth}
                >{\raggedright\arraybackslash}p{0.24\linewidth}
                >{\raggedright\arraybackslash}p{0.26\linewidth}}
\toprule
Mutation class & Ignorability & Mechanism of violation & Condition for conditional ignorability \\
\midrule
Point mutations (normal replication) & Satisfied & --- & --- \\
CpG / transition-biased substitutions & Weakly violated & Positional mutation-rate heterogeneity & Condition on local genomic context \\
Indels and gene duplications & Approximately satisfied & Repetitive-sequence and length bias & Condition on genomic architecture \\
Stress-induced mutagenesis (SOS response) & Violated & Mutation rate depends on fitness via stress & Condition on stress state and covariates \\
Transposable element insertions & Partially violated & Chromatin-dependent insertion preference & Condition on chromatin state and TE family \\
Experimental mutagenesis & Satisfied by design & --- & Randomisation in design \\
\bottomrule
\end{tabular}
\end{table}

\subsection{Point Mutations Under Normal Replication}
\label{sec:point}

The foundational experiment of Luria and Delbr\"{u}ck established that spontaneous mutations in bacteria arise prior to and independently of the selective conditions that reveal them (Luria \& Delbr\"{u}ck, 1943). This result, confirmed across prokaryotic and eukaryotic systems, provides the empirical cornerstone of the ignorability assumption: under normal DNA replication, the probability of a base-pair substitution at a given site is governed by polymerase fidelity, proofreading, and mismatch repair, none of which have access to information about the phenotypic or fitness consequences of the error.

Typical per-base-pair, per-generation error rates lie in the range $10^{-8}$ to $10^{-10}$ for organisms with intact repair machinery (Drake et al., 1998, Lynch, 2010). The identity of the substitution (e.g. $\text{A}\to \text{G}$ versus $\text{A}\to \text{T}$) is determined by thermodynamic base-pairing properties and polymerase active-site geometry, not by downstream phenotypic effects. Therefore, for this class of mutations, there is no back-door path from potential fitness outcomes to mutation assignment, and ignorability holds unconditionally: the treatment assignment $M$ is independent of the potential outcomes $\{Y(0), Y(1)\}$ without the need to condition on additional covariates.

\subsection{Transition--Transversion Bias and CpG Hypermutability}
\label{sec:cpg}

Point mutations are blind to their fitness effects but not uniformly distributed across the genome. Two well-characterised biases deserve attention. First, transitions (purine $\leftrightarrow$ purine or pyrimidine $\leftrightarrow$ pyrimidine) often occur at higher rates than transversions, owing to the greater thermodynamic similarity of same-class bases and the geometry of wobble base pairs in the polymerase active site (Lynch, 2007). Second, in organisms with cytosine methylation, CpG dinucleotides exhibit mutation rates 10- to 50-fold higher than the genomic average because spontaneous deamination of 5-methylcytosine produces thymine, a normal base that can evade mismatch repair (Coulondre et al., 1978; Bird, 1980).

These biases introduce heterogeneity in the probability of mutation across genomic positions. Recent work on mutation bias and adaptation has shown that while transition--transversion bias can influence which adaptive substitutions are observed, this reflects a bias in the supply of mutations, not a dependence of the occurrence of a specific mutation at a given site on its fitness effect (Stoltzfus \& McCandlish, 2017; Cano \& Payne, 2020). The distinction is essential: genomic context $X$ affects both baseline fitness $Y(0)$ and mutation probability $M$, creating a back-door path $Y(0) \leftarrow X \rightarrow M$, but conditional on $X$, the occurrence of a given base substitution carries no information about the sign or magnitude of its fitness effect. Conditioning on local sequence context (e.g. CpG status, base identity, neighbouring nucleotides) thus restores ignorability:
\[
\{Y(0), Y(1)\} \perp\!\!\!\perp M \mid X,
\]
and Theorem~\ref{thm:identification} applies to the context-conditional AME $E[Y(1)-Y(0)\mid X]$.

\subsection{Insertions, Deletions, and Gene Duplications}
\label{sec:indels}

Small insertions and deletions (indels) arise primarily through replication slippage at repetitive sequences and through template switching during DNA synthesis, while gene duplications and larger structural variants result from unequal crossing-over, retrotransposition, and segmental duplication (Zhang et al., 2003; Hastings et al., 2009). For all of these classes, the probability of occurrence depends on local genomic architecture --- repetitive element density, sequence composition, and replication fork dynamics --- rather than on the phenotypic consequences of the event.

As with CpG bias, positional heterogeneity in indel and duplication rates constitutes a violation of unconditional ignorability that is correctable by conditioning on measurable genomic features such as repeat content, microsatellite length, and proximity to replication origins or recombination hotspots. Once genomic architecture $X$ is accounted for, the occurrence of an indel or duplication event at a given position carries no information about whether that event will be beneficial, neutral, or deleterious, and conditional ignorability is restored. The relevant estimand is again a conditional AME $E[Y(1)-Y(0)\mid X]$ that can be integrated over the distribution of $X$ to obtain unconditional AMEs for use in evolutionary dynamics.

An additional subtlety arises for gene duplications: duplicated genes often experience relaxed purifying selection and may accumulate mutations at elevated rates over evolutionary time (neofunctionalisation or subfunctionalisation) (Lynch 2007). This is, however, a post-treatment process that affects the fitness trajectory of the duplicated gene over subsequent generations, not the probability of the initial duplication event. As such, it does not violate the ignorability assumption as formulated here, which concerns only the independence of treatment assignment from potential outcomes at the time of the mutational event.

\subsection{Stress--Induced Mutagenesis}
\label{sec:stress}

A more consequential violation of ignorability arises from stress-induced mutagenesis. When bacteria and other organisms encounter environmental stress---starvation, antibiotic exposure, oxidative damage, DNA double-strand breaks---stress responses such as the bacterial SOS system are activated, leading to upregulation of error-prone DNA polymerases and suppression of error-correcting pathways (Radman, 1975; Foster, 2007; Fitzgerald et al., 2017). The result is a transient increase in mutation rate, often by one to two orders of magnitude, tightly coupled to the organism's physiological state.

From the perspective of causal inference, the key feature is that mutation assignment is no longer independent of potential fitness: organisms experiencing stress are typically in a state of reduced baseline fitness (they are poorly adapted to the current environment or are suffering damage), and the same stress response increases their mutation rate. This creates a confounding pathway
\[
Y(0) \leftarrow \text{Stress} \rightarrow M,
\]
so that $P(M=1)$ becomes correlated with $Y(0)$, the counterfactual fitness in the absence of mutation, even before any causal effect of the new mutation is realised. Ignorability is violated because the mutation process ``knows'', via stress, that it is acting on low-fitness individuals.

Several features of stress-induced mutagenesis determine the severity and tractability of this violation. First, in many microbial experiments, stress is approximately homogeneous within a population or subpopulation: in a clonal culture exposed uniformly to an antibiotic, all cells experience similar elevation of mutation rates, so within the stressed subpopulation ignorability may hold \emph{conditional on stress status}. This suggests a stratification approach in which we define a covariate $S \in \{\text{stressed},\text{unstressed}\}$ and assume
\[
\{Y_i(0), Y_i(1)\} \perp\!\!\!\perp M_i \mid \mathbf{X}_i, S_i,
\]
where $\mathbf{X}_i$ collects additional genomic and phenotypic covariates. If stress status is observable (e.g. via reporters of SOS activation or growth-rate measurements), then conditional ignorability holds within strata of $(\mathbf{X},S)$, and Theorem~\ref{thm:identification} applies to conditional AMEs such as $E[Y(1)-Y(0)\mid \mathbf{X},S=\text{stressed}]$.

Second, stress--induced mutagenesis is typically transient with mutation rates returning to baseline once the stress is alleviated, resulting in a time-limited violation which affects only a fraction of mutational events in a lineage's history (Foster, 2007; Fitzgerald et al., 2017). Third, when stress status cannot be measured directly, the framework can still accommodate the violation using methods from longitudinal observational studies, such as inverse-probability-of-treatment weighting based on observed covariate histories or sensitivity analysis that quantifies how strong the unmeasured confounding through stress would need to be to overturn qualitative conclusions.

In sum, stress-induced mutagenesis provides a biologically important example in which ignorability fails unconditionally but can often be restored, at least approximately, by conditioning on stress state and related covariates. In that sense, the ``remedies'' in Table~\ref{tab:ignorability} specify covariate sets under which mutation events can be treated as quasi--random interventions, extending the causal--evolutionary framework to this class of mechanisms.

\subsection{Transposable Element Insertions}
\label{sec:te}

Transposable elements (TEs) constitute a distinct mutation class because their insertion sites are influenced by interactions between the transposition machinery and the chromatin environment. Different TE families exhibit different insertion preferences: some retrotransposons insert upstream of particular classes of genes, others target open chromatin or promoter-proximal regions, and still others are enriched in gene-rich, GC-rich genomic compartments (Fescotte \& Pritham, 2007; Slotkin \& Martienssen, 2007). These insertion preferences are mediated by binding of TE integrases or transposases to specific chromatin features and protein complexes.
These preferences create a non-uniform distribution of TE-mediated mutations across the genome that is correlated with functional density and expression. Regions of open chromatin, where TE insertion is more likely, are also regions enriched in expressed genes where mutations may have larger fitness effects. This constitutes a partial violation of ignorability because the probability of a TE insertion at a given site carries information about the likely magnitude (though not necessarily the sign) of its fitness effect. Indeed,  both are functions of underlying chromatin and genomic context. In graphical terms, chromatin state $X$ is a common cause of both $M$ and $Y(0)$, opening a back-door path $Y(0) \leftarrow X \rightarrow M$.

As with other context-dependent biases, the violation is correctable in principle through conditioning. The relevant covariates --- chromatin accessibility (assayable by ATAC-seq or DNase-seq), histone modifications, TE family identity, distance to the nearest gene, and local expression level --- are measurable in modern genomic datasets (Rebollo 2012). Conditional on these features $X$, the specific fitness effect of a TE insertion at a given site is independent of the insertion probability, restoring conditional ignorability:
\[
\{Y(0), Y(1)\} \perp\!\!\!\perp M \mid X.
\]
The corresponding estimands are TE- and context-specific AMEs $E[Y(1)-Y(0)\mid X]$, which can be aggregated over the empirical distribution of $X$ to obtain unconditional contributions of TE insertions to evolutionary dynamics.

TE mobilisation can also be influenced by stress, creating an overlap with the stress-induced mutagenesis category (Section~\ref{sec:stress}). Environmental stressors can derepress TEs through relaxation of epigenetic silencing, linking TE activity to the organism's fitness state (Capy et al., 2000; Lanciano \& Mirouze, 2018). In such cases, stress depresses baseline fitness and elevates the probability of a mutational event, creating an analogous confounding structure, and both chromatin context and stress state enter the conditioning set needed to restore ignorability.

\subsection{Experimental Mutagenesis}
\label{sec:experimental}

In experimental evolution and functional genomics, mutations are often introduced by the investigator through chemical mutagenesis, UV irradiation, random transposon insertion libraries, or CRISPR-based tiling screens. When the mutagenic agent is applied uniformly and the mutations are distributed across the genome without regard to their fitness effects, ignorability is satisfied by design, precisely as randomisation ensures ignorability in a clinical trial. Deviations from perfect randomness, for example context--dependent CRISPR cutting efficiency, can be treated analogously to the context biases discussed above.

The practical implication is that experimental evolution systems provide convenient settings for estimating AMEs. Mutation--accumulation experiments, in which lines are propagated through repeated single--cell bottlenecks to minimise selection, directly estimate the distribution of fitness effects of new mutations under conditions where both ignorability and the absence of selection are enforced experimentally (Halligan \& Keightley, 2009).

\subsection{Implications for the Causal Framework}
\label{sec:ignor_implications}

The dominant source of new mutations in many organisms --- spontaneous point substitutions during normal DNA replication --- satisfies ignorability unconditionally. The principal sources of correctable violations (sequence-context bias, CpG hypermutability, indel and duplication hotspots, TE insertion preferences, stress-linked TE mobilisation) are all addressable, in principle, through conditioning on measurable genomic and physiological covariates. The most serious violation, stress-induced mutagenesis, is biologically important but temporally restricted and often amenable to stratification by stress status.

The parallel with observational studies in epidemiology is instructive. In a well--designed observational study, treatment assignment (e.g. smoking status) is not randomised, but conditional on a rich set of covariates (age, sex, socioeconomic status, comorbidities), the treatment can be considered approximately ignorable. Similarly, mutation assignment in natural populations is not perfectly randomised, but conditional on genomic context and organismal stress state, the residual confounding is often small relative to the causal signal of selection (Hern\'{a}n \& Robins, 2006).

This implies that the relevant causal estimands are often \emph{conditional} Average Mutation Effects (AMEs), defined given covariates summarising genomic context or stress state, which can then be aggregated over the empirical distribution of these covariates to recover the unconditional AMEs and Total Mutation Effects that enter the dynamical equations of Section~\ref{sec:dynamics}. In each case where ignorability is marked as weakly or partially violated in Table~\ref{tab:ignorability}, the condition listed in the final column specifies a covariate set under which Assumption~\ref{as:ignorability} is restored. Under these conditions, Theorem~\ref{thm:identification} continues to apply, but to conditional AMEs rather than unconditional ones.

Although I have described these adjustments using the language of causal analysis of observational data (conditioning, weighting, stratification), their primary role remains theoretical in the sense that they specify the conditions under which mutation events can be treated as exogenous interventions in the causal sense, thereby extending the scope of our causal--evolutionary framework to a broad range of mutational mechanisms.

\newpage
\section*{References}

\begin{description}
\item Bird, A.\ P.,\ (1980). DNA methylation and the frequency of CpG in animal DNA. \emph{Nucleic Acids Res.}, 8, 1499--1504.
\item Cano, A.\ V., \& Payne, J.\ L.\ (2020). Mutation bias interacts with composition bias to influence adaptive evolution. \emph{PLoS Comput.\ Biol.}, 16, e1007523.
\item Capy, P. et al.,\ (2000). Stress and transposable elements. \emph{Genetica}, 107, 149--159.
\item Charlesworth, B.,\ (1994). \emph{Evolution in Age-Structured Populations} (2nd ed.). Cambridge.
\item Coulondre, C., et al.\ (1978). Molecular basis of base substitution hotspots in \emph{E.\ coli}. \emph{Nature}, 274, 775--780.
\item Doebeli, M., et al., \ (2017). Towards a mechanistic foundation of evolutionary theory. \emph{eLife}, 6, e23804.
\item Drake, J.\ W., \ (1991) A constant rate of spontaneous mutation in DNA-based microbes. \emph{Proc. Natl. Acad. Sci. USA.}, 88, 7160--7164
\item Drake, J.\ W., et al. \ (1998). Rates of spontaneous mutation. \emph{Genetics}, 148, 1667--1686
\item Eigen, M., \& Schuster, P.\ (1977). The hypercycle. \emph{Naturwissenschaften}, 64, 541--565.
\item Fitzgerald, D.\ M., et al.\ (2017). What is mutation? A chapter in the series: How microbes ``jeopardize'' the modern synthesis. \emph{PLoS Genet.}, 13, e1007003.
\item Foster, P.\ L.\ (2007). Stress-induced mutagenesis in bacteria. \emph{Crit.\ Rev.\ Biochem.\ Mol.\ Biol.}, 42, 373--397.
\item Feschotte, C., \& Pritham, E.\ J.\ (2007). DNA transposons and the evolution of eukaryotic genomes. \emph{Annual Review of Genetics}, 41, 331--368.
\item Gillespie, J.\ H.\ (1973). Natural selection with varying selection coefficients --- a haploid model. \emph{Genet.\ Res.}, 21, 115--120.
\item Gould, S.\ J., \& Lewontin, R.\ C.\ (1979). The spandrels of San Marco. \emph{Proc.\ R.\ Soc.\ Lond.\ B}, 205, 581--598.
\item Halligan, D.\ L., \& Keightley, P.\ D.\ (2009). Spontaneous mutation accumulation studies in evolutionary genetics. \emph{Annu.\ Rev.\ Ecol.\ Evol.\ Syst.}, 40, 151--172.
\item Hamilton, W.\ D.\ (1964). The genetical evolution of social behaviour. \emph{J.\ Theor.\ Biol.}, 7, 1--52.
\item Hastings, P.\ J.\ et al. (2009). DNA repair: A link between cancer and aging. \emph{DNA Repair}, 8, 1258--1267.
\item Hern\'{a}n, M.\ A., \& Robins, J.\ M.\ (2006). Estimating causal effects from epidemiological data. \emph{Journal of Epidemiology and Community Health}, 60(7), 578--586.
\item Holland, P.\ W.\ (1986). Statistics and causal inference. \emph{JASA}, 81, 945--960.
\item Hudgens, M.\ G., \& Halloran, M.\ E.\ (2008). Toward causal inference with interference. \emph{JASA}, 103, 832--842.
\item Lanciano, S., \& Mirouze, M.\ (2018). Transposable elements: all mobile, all different and emerging actors. \emph{Curr.\ Opin.\ Plant Biol.}, 42, 44--51.
\item Lenski, R. et al.\ (1991). Long-term experimental evolution in \emph{Escherichia coli}. I. Adaptation and divergence during 2,000 generations. \emph{Am.\ Nat.}, 138, 1315--1341.
\item Lewontin, R.\ C., \& Cohen, D.\ (1969). On population growth in a randomly varying environment. \emph{Proc.\ Natl.\ Acad.\ Sci.\ USA}, 66, 1056--1060.
\item Luria, S.\ E., \& Delbr\"{u}ck, M.\ (1943). Mutations of bacteria from virus sensitivity to virus resistance. \emph{Genetics}, 28, 491--511.
\item Lynch, M.\ (2007). The frailty of adaptive hypotheses. \emph{PNAS}, 104, 8597--8604.
\item Lynch, M.\ (2010). Rate, Molecular Spectrum, and Consequences of Human Mutation. \emph{Proceedings of the National Academy of Sciences}, 107, 961--968
\item Maynard Smith, J., \& Price, G.\ R. (1973) The logic of animal conflict \emph{Nature}, 246, 15--18.
\item Metz, J.\ A.\ J., et al.\ (1992). How should we define ``fitness''? \emph{TREE}, 7, 198--202.
\item Otsuka, J.\ (2016). Causal foundations of evolutionary genetics. \emph{BJPS}, 67, 247--269.
\item Pearl, J.\ (2010). Causal inference. \emph{Causality: objectives and assessment}, 18, 39--58.
\item Radman, M.\ (1975). SOS repair hypothesis. In \emph{Molecular Mechanisms for Repair of DNA}, pp.\ 355--367. Plenum.
\item Rebollo, J.\ E., et al. (2012). Transposable elements: An ancient and diverse source of genetic innovation. \emph{Nature Reviews Genetics}, 13, 491--502.
\item Robins, J.\ M.\ (1986). A new approach to causal inference in mortality studies. \emph{Math.\ Model.}, 7, 1393--1512.
\item Rosenbaum, P.\ R.\ (2002). \emph{Observational Studies} (2nd ed.). Springer.
\item Rubin, D.\ B.\ (1974). Estimating causal effects of treatments. \emph{J.\ Educ.\ Psychol.}, 66, 688--701.
\item Slotkin, D.\ J., \& Martienssen, R.\ (2007). Transposable elements and the epigenetic regulation of the genome. \emph{Nature Reviews Genetics}, 8, 272--285
\item Stoltzfus, A., \& McCandlish, D.\ M.\ (2017). Mutational biases influence parallel adaptation. \emph{Mol.\ Biol.\ Evol.}, 34, 2163--2172.
\item Tenaillon, O., et al.\ (2016). Tempo and mode of genome evolution in a 50,000-generation experiment. \emph{Nature}, 536, 165--170.
\item Tchetgen Tchetgen, E.\ J., \& VanderWeele, T.\ J.\ (2012). On causal inference in the presence of interference. \emph{Stat.\ Methods Med.\ Res.}, 21, 55--75.
\item Wade, M. \ J. (1978) A critical review of the models of group selection. \emph{Quarterly Review of Biology}, 53, 101--114.
\item Zhang, F., et al.\ (2003). Segmental duplication and the evolution of the primate genome. \emph{Nature Reviews Genetics}, 4, 761--773.
\end{description}

\end{document}